\documentclass[conference]{IEEEtran}
\IEEEoverridecommandlockouts
\pdfoutput=1 
\usepackage{cite}
\usepackage{amsmath,amssymb,amsfonts}
\usepackage{algorithmic}
\usepackage{graphicx}
\usepackage{textcomp}
\usepackage{xcolor}
\usepackage{float}
\usepackage{amsthm}
\usepackage{enumitem}
\usepackage{tikz}
\usetikzlibrary{automata, positioning}
\newtheorem{lemma}{Lemma}

\def\BibTeX{{\rm B\kern-.05em{\sc i\kern-.025em b}\kern-.08em
    T\kern-.1667em\lower.7ex\hbox{E}\kern-.125emX}}
\begin{document}

\title{Denial-of-Service Attacks on C-V2X Networks\\
}

\author{\IEEEauthorblockN{Nata\v{s}a Trkulja}
\IEEEauthorblockA{\textit{ECE Department} \\
\textit{Boston University}\\
Boston, MA, USA \\
ntrkulja@bu.edu}
\and
\IEEEauthorblockN{David Starobinski}
\IEEEauthorblockA{\textit{ECE Department} \\
\textit{Boston University}\\
Boston, MA, USA \\
staro@bu.edu}
\and
\IEEEauthorblockN{Randall A. Berry}
\IEEEauthorblockA{\textit{ECE Department} \\
\textit{Nothwestern University}\\
Evanston, IL, USA \\
rberry@northwestern.edu}
}

\maketitle

\begin{abstract}
Cellular Vehicle-to-Everything (C-V2X) networks are increasingly adopted by automotive original equipment manufacturers (OEMs).
C-V2X, as defined in 3GPP Release 14 Mode 4, allows vehicles to self-manage the network in absence of a cellular base-station. Since C-V2X networks convey safety-critical messages, it is crucial to assess their security posture.   
This work contributes a novel set of Denial-of-Service (DoS) attacks on C-V2X networks operating in Mode 4. The attacks are caused by  adversarial resource block selection and vary in sophistication and efficiency.  In particular, we consider ``oblivious'' adversaries that ignore recent transmission activity on resource blocks, ``smart'' adversaries that do monitor activity on each resource block, and ``cooperative'' adversaries that work together to ensure they attack different targets.
We analyze and simulate these attacks to showcase their effectiveness. Assuming a fixed number of attackers, we show that at low vehicle density, smart and cooperative attacks can significantly impact network performance, while at high vehicle density, oblivious attacks are almost as effective as the more sophisticated attacks. \\
\end{abstract}


\section{Introduction}
Cellular Vehicle-to-Everything (C-V2X) is an LTE-based technology that enables communications between automotive vehicles (V2V), vehicles and pedestrians (V2P), vehicles and infrastructure (V2I), and vehicles and the network (V2N). Specifications for C-V2X communications came out with release 14 of the 3rd Generation Partnership Project (3GPP)~\cite{Mansouri, TSGR2017}. The release defined two new modes of LTE operation (Mode 3 and Mode 4) whose major difference is resource allocation: where Mode 3 relies on the cellular base-station to perform resource allocation, Mode 4 has been designed to allow vehicles to allocate resources on their own. Operating in Mode 4, vehicles utilize their own radio user equipment (UE) to communicate with one another without having access to the cellular network. This along with having an extended range, high reliability, and pre-existing infrastructure 
are what makes C-V2X a very attractive technology~\cite{QualcommTechnologies2019}. Thus, original equipment manufacturers (OEMs) such as Ford, Audi, BMW, Mercedes Benz, and the entire 5G Automotive Association (5GAA) have greatly increased their C-V2X research efforts \cite{5GAutomotiveAssociation2020},\cite{Hill2019}. This sets C-V2X as a leading technology for vehicular communications alongside Dedicated Short Range Communication (DSRC). As with designing any other system on a vehicle, the most important aspect of designing a vehicular communication system such as an on-board C-V2X system is safety. Safety, in this case, cannot be guaranteed without security as adversarial attacks on vehicular networks can have tremendous safety consequences. This is why it is vital to perform research regarding security of C-V2X networks, especially those operating in Mode 4 being potentially more vulnerable than those operating in Mode 3, due to the lack of cellular infrastructure for monitoring the network.

To operate in Mode 4, vehicles select appropriate resource blocks (RBs) that will be used to transmit their Basic Safety Messages (BSMs). A resource block is a set of OFDM subcarriers within a given time-slot. Vehicles operating in Mode 4 sense and process incoming signals to select resource blocks from the 20\% of those available with the lowest Received Signal Strength Indicator (RSSI) \cite{TSGR2017}. When a vehicle selects a resource block, it periodically transmits BSMs on that resource block for a certain duration called the ``semi-persistent period''. Once the semi-persistent period expires, the vehicle selects a new resource block with a probability $p$\footnote{In general $p$ can be set anywhere in the range between 0.2 and 1. In our simulations, we use $p=0.2$ as in~\cite{Wang2018},  while in some of our theoretical results we allow $p$ to take any value in $[0,1]$.}. As vehicles go through this process, 
they may randomly end up choosing the same resource block for a period of time which results in packet collisions assuming resource blocks are orthogonal. It is these types of collisions that serve as the basis for the Denial-of-Service (DoS) attacks on C-V2X networks that we are envisioning.



The purpose of this paper is to investigate if and how the resource block selection process can be abused by a malicious party. Specifically, we consider one or more adversaries, each equipped with a C-V2X device, that aim to launch a DoS attack on the network by inducing packet collisions and hence lowering the packet reception ratio (PRR). We consider adversaries with different levels of sophistication. Thus, we consider ``oblivious'' adversaries, which ignore recent transmission activity on resource blocks and just select a new resource block with probability $p' \not = p$, versus  ``smart'' adversaries that monitor activity on each resource block before deciding on which one to transmit next. We also consider ``cooperative'' adversaries that work together to ensure they attack different targets. 
Our goal is to assess the potency of these attacks as a function of the vehicle density and the number of attackers.

Our main contributions are:
\begin{enumerate}
    \item We show that C-V2X networks are vulnerable to DoS attacks caused by adversarial resource block selection. 
    \item We introduce attack types of increasing level of sophistication, and investigate their effect through analysis and simulation based on the specifications of the C-V2X protocol.
    \item We show that the most effective oblivious attack is one where the attacker selects a new resource block with probability $p'=1$. This result is formally proven for a special case and shown to hold in general through simulation.
    \item We show that collaborative attacks are more potent than smart attacks, which in turn are more potent than oblivious attacks, especially at low vehicle density. Yet, at high vehicle density, oblivious attacks become almost as effective as the more sophisticated attacks. 
    \item We show that attacking every transmission period is more effective than attacking every semi-persistent period, although not in a significant way. The former type of attack is easier to detect, however, as it is evidently non-compliant with the protocol.

\end{enumerate}

To the best of our knowledge, the attacks presented in this paper have not been studied before. The goal of the paper is to provide insight into these attacks, and hence we start by evaluating them in simple network configurations. In the paper's conclusion, we discuss potential areas for future work. 

This paper is organized as follows. Section II discusses the existing work regarding DoS attacks on vehicular networks and C-V2X security in general. Section III describes the attack types that we developed. Section IV presents  analysis that sheds light on the effectiveness of the various attack types. Section V showcases the results we obtained by simulation to investigate the attack types' effectiveness under a variety of conditions. Section VI concludes the paper. 
\section{Related Work}
Denial-of-Service attacks have long been recognized as a significant threat to vehicular networks \cite{Hasbullah2010}. Their ability to congest the RF spectrum and prevent vehicles from accessing necessary RF resources, their ability to impede the flow of safety-critical information between vehicles, as well as their ability to deny vehicles access to road-side units (RSUs) have rendered DoS attacks as one of the most dangerous attacks against vehicular networks.

DoS attacks have been studied to a great extent within the context of DSRC. The study performed in \cite{Punal2015} evaluated the performance of DSRC-based vehicular networks that are under jamming attacks. Jamming attacks are a form of DoS attacks where jammers interfere with legitimate signals to prevent receivers from properly demodulating them. The study investigated the effects of constant, periodic, and reactive jamming of DSRC devices in an anechoic chamber, where reactive jamming only occurs if the attacker senses that an energy threshold on a certain band has been exceeded. The work showed that through periodic jamming signal-to-interference-plus-noise (SINR) ratio is impaired by 56 dB. 
Similarly, the work presented in \cite{Basciftci2016} evaluates the vulnerability of DSRC-based vehicular networks to three types of jamming attacks: flat, random, and smart. Flat jamming involves jamming all resource blocks with equal power, and smart jamming consists of sniffing for preambles of packets in transmission and then jamming only resource blocks with detected activity. They show through experimentation that complete blockage is possible with smart jamming even with low powers levels used by the jammer. Having recognized the impact of DoS attacks on vehicular networks, recent research work has focused on developing DoS detection mechanisms and antijamming techniques \cite{Benslimane2017,Gu2018,Lyamin2019}. 

Being more recent than DSRC, C-V2X based research has mostly focused on the performance of C-V2X networks, e.g.~\cite{Nabil2018,Wang2018}. There is a limited pool of work that investigates the security aspect of C-V2X networks. The work in \cite{Marojevic2018} reviews potential threats to C-V2X networks, including Denial-of-Service attacks, recognizing that such attacks have the potential to compromise the reliability of C-V2X service. The work in~\cite{Ahmed2018} evaluates the security of the C-V2X protocol as outlined in~\cite{TSGR2017} and proposes a privacy-preserving scheme which consists of key distribution and management, V2X Domain Registration, V2X Service Registration, and V-UE Pseudonym Registration. Additional work focused on developing a mechanism for detecting DoS attacks while operating in Mode 3 of the C-V2X protocol \cite{Li2018}. This work 
identifies a DoS attack that maliciously reserves the resources at an evolved NodeB (eNB) node with the goal of denying service to honest vehicles. To the best of our knowledge, our work is the first to explore the impact of DoS attacks in Mode 4 of C-V2X networks.
\section{Attack Types}
In this section, we introduce different types of DoS attacks based on adversarial resource block selection. These attacks increase the likelihood of packet collisions. 
Under the assumption that all resource blocks are orthogonal, packet collisions happen when two or more vehicles in the system choose the same resource block at the same time. Here, we assume that all vehicles are within a close enough range of each other so that when such a collision occurs, this results in neither of the two vehicles' messages being received by the remaining vehicles. Note that such collisions may occur even in the absence of attacks due to the random nature of resource block selection. There simply will be instances when two or more vehicles happen to randomly choose the same resource block even without any malicious vehicles in the system. This scenario, without malicious vehicles, represents the \textbf{baseline} case against which we will be comparing the different attack types. 

Before describing the attack types, we need to define several important system parameters. We assume that BSMs are sent by each vehicle every $T_{tr}$ seconds, which we refer to this as the \emph{transmission period}. Each vehicle $v_i$ uses a resource block $r_i$ to transmit its BSM,  where $i$ represents a specific resource block. $T_s$ is the number of transmission periods inside a semi-persistent period, that is vehicle $i$ transmits BSMs on resource block $r_i$ for $T_s$ times. Upon the expiration of the semi-persistent period, a honest vehicle changes its resource block with a probability $p$ as mentioned earlier.


\textbf{Attacker goal}. The attacker's goal is to minimize the packet reception ratio (PRR) in the network, which is the fraction of correctly received packets to the total number of transmitted packets. 


\textbf{Attacker capabilities}. We assume each attacker to be a vehicle with the same C-V2X capabilities as the target vehicles. In particular, the attacker can only use one resource block at a time. Depending on the attacker type, we assume that it is able to switch resource blocks after each transmission period or after each semi-persistent period. We define the \emph{attack period}, $T_a$, as the time during which the attacker does not change its resource block. 

We use the following characteristics to further classify our attacks:

\begin{itemize}
\item \textbf{Oblivion}: An attack is oblivious if attacker vehicles do not have to listen in on target vehicles' communication (i.e., which resource blocks they are transmitting on).
\item \textbf{Cooperation}: An attack type is cooperative if attackers cooperate with one another to deliver a more efficient attack.
\item \textbf{Deniability}: An attack is deniable if it cannot be decisively proven that the attacker is misbehaving. This is likely due to the attack meeting all communication protocol requirements. One such example would be an attack type with attack period equal to the semi-persistent period (i.e., the attacker only selects a new resource block at the end of the semi-persistent period). While a node is not supposed to switch to a busy resource block, an attacker could claim that its device did not sense energy on that resource block. 
A clearly non-deniable attack is one with an attack period that is smaller than the semi-persistent period, e.g., an attack period equal to the transmission period.
\end{itemize}


\subsection{Attack Type 1: Oblivious Attack}
In \textbf{attack type 1}, attacker vehicles select a resource block with a probability $p'$ choosing from the entire pool of RBs. There are a couple of special cases of this attack type: 

\begin{itemize}
\item $p' = 0$. This implies that the attacker vehicles never change the resource block they originally and randomly selected. As the attacker resource block selection is random, it may or may not select a resource block that is in use by a target vehicle at the start of the attack. Attacker vehicles may also choose the same resource block, even if it is an idle resource block. Therefore, the collisions are solely dependent on the pure chance that the attacker vehicle selects a resource block already in use at the start of the attack.  
This special case is represented in Figure \ref{fig:model1_scP0}. 

\item $p' = 1$. This implies that the attacker vehicles change the resource block upon expiration of every attack period, $T_{a}$. They randomly select a new resource block among all possible blocks. The attacker vehicles may or may not choose a resource block that is already in use; this is again dependent on chance. We anticipate this special case to be more effective than when $p' = 0$. 
This special case is represented in Figure \ref{fig:model1_scP1}.
\end{itemize} 

\begin{figure}[t]
\centering
\includegraphics[scale=0.6]{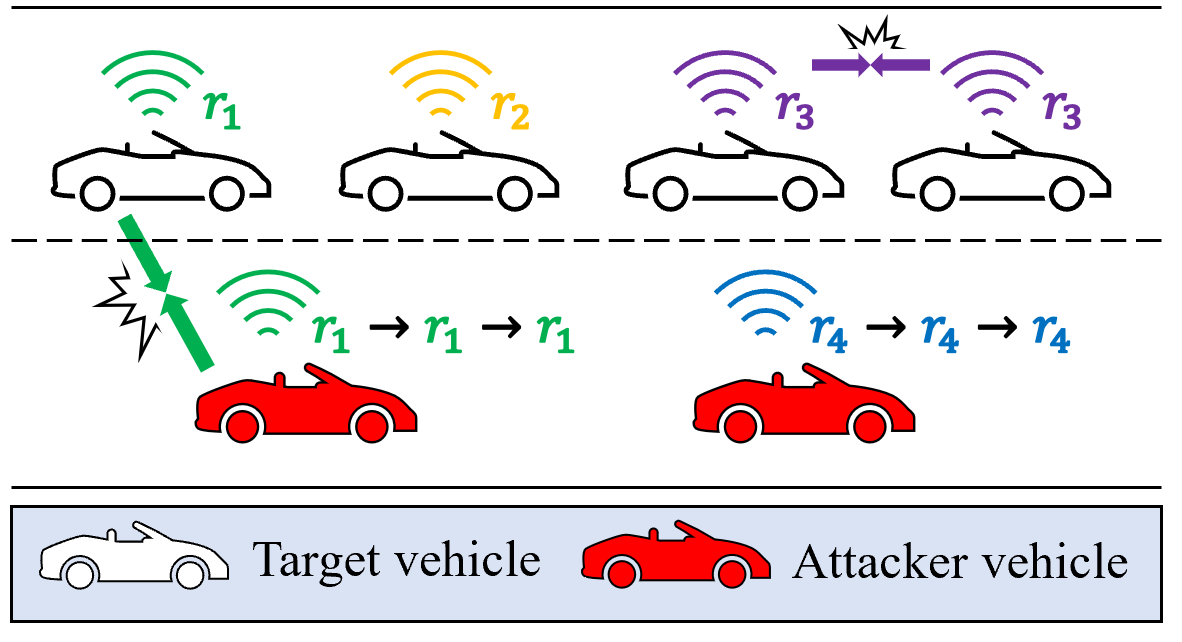}
\caption{Attack Type 1 (Special Case: $p' = 0$). The figure depicts an instance of the special case where one attacker randomly chooses a resource block that is already in use, making it a successful attack until the target changes its resource block. The other attacker randomly chooses a resource block that is not in use. This attacker will never succeed in its attack as no target vehicle will choose that resource block as long as the attacker is using it.}
\label{fig:model1_scP0}
\end{figure}

\begin{figure}[t]
\centering
\includegraphics[scale=0.6]{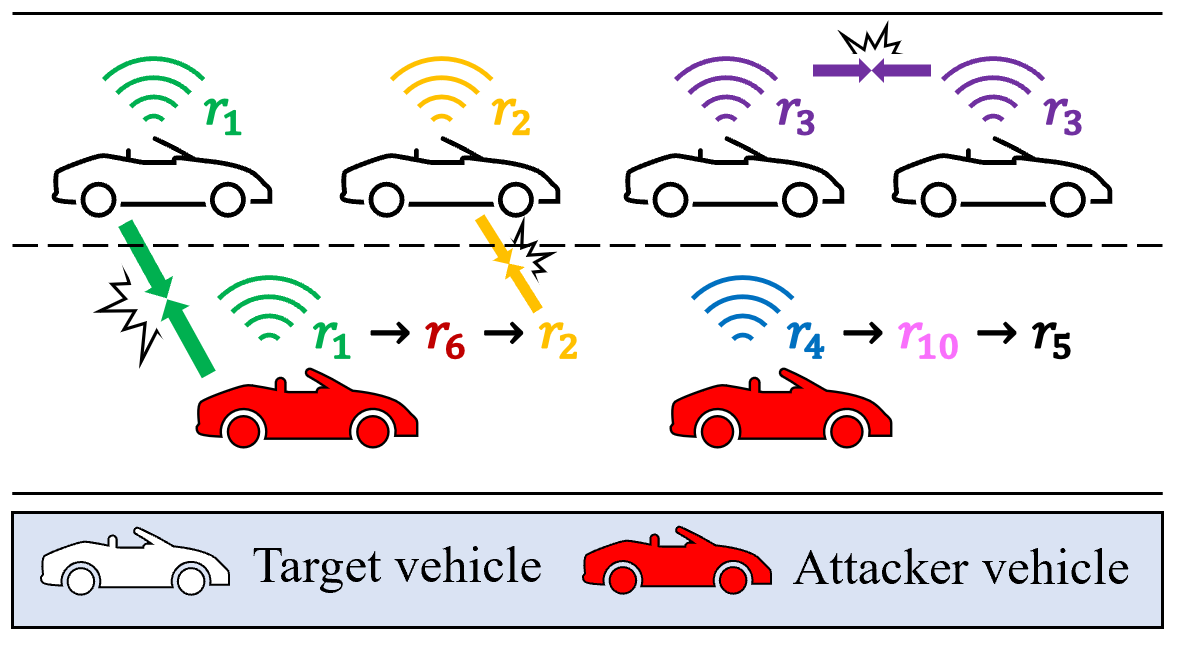}
\caption{Attack Type 1 (Special Case: $p' = 1$). The figure depicts an instance of the special case where one attacker randomly chooses a resource block that is already in use, making it a successful attack until the attacker changes its $r_i$ in the next semi-persistent period. The same attacker will choose another used $r_i$ in the third semi-persistent period, having completed two successful attacks. The other attacker will not choose a used $r_i$ in the first three semi-persistent periods and there will be no collisions.}
\label{fig:model1_scP1}
\end{figure}

\subsection{Attack Type 2: Smart Attack}
In \textbf{attack type 2}, attacker vehicles look for resource blocks that are being used by only one vehicle (which we refer to as a \textit{loner vehicle}) and select one of such resource blocks. The adversaries have to spend a period of time listening to messages sent by target vehicles (non-oblivious) to establish a list of loner vehicles' resource blocks from which they randomly choose their own. Attacker vehicles do not cooperate implying that they may choose the same loner vehicle to attack during one attack period, reducing the effectiveness of attack. This attack type is more difficult to implement compared to the oblivious one.
This scenario is represented in Figure \ref{fig:model2}.

\begin{figure}[t]
\centering
\includegraphics[scale=0.6]{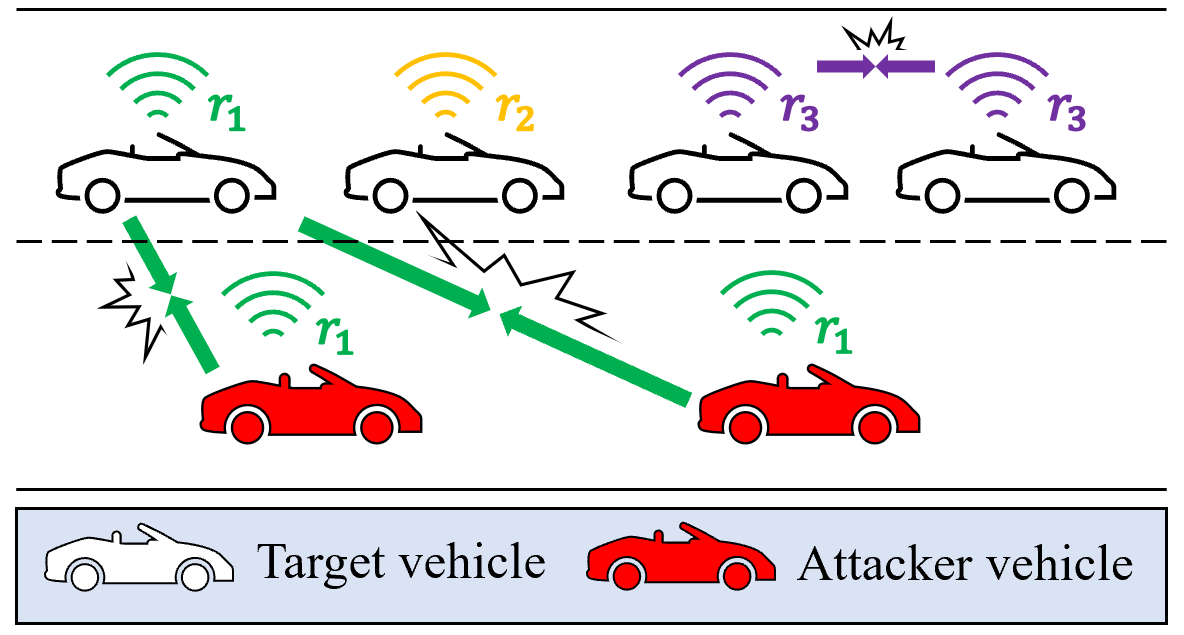}
\caption{Attack Type 2. The figure depicts an instance when the attacker vehicles have successfully chosen a loner vehicle to attack, but have both chosen the same loner vehicle. This is because of their non-cooperation.}
\label{fig:model2}
\end{figure}

\subsection{Attack Type 3: Cooperative Attack}
In \textbf{attack type 3}, attacker vehicles look for all resource blocks that are being used by loner vehicles and each selects one of those resource blocks while ensuring no two attackers choose the same block. They have to spend a period of time listening to messages sent by target vehicles (non-oblivious) to establish a list of loner vehicles' resource blocks from which they choose their own RB. Additionally, they have to communicate between themselves to ensure that they all select a resource block from this list that is different from one another. This attack type is cooperative making it highly efficient in theory. It would be more challenging to implement in reality, however. 
This scenario is represented in Figure \ref{fig:model3}.

\begin{figure}[t]
\centering
\includegraphics[scale=0.6]{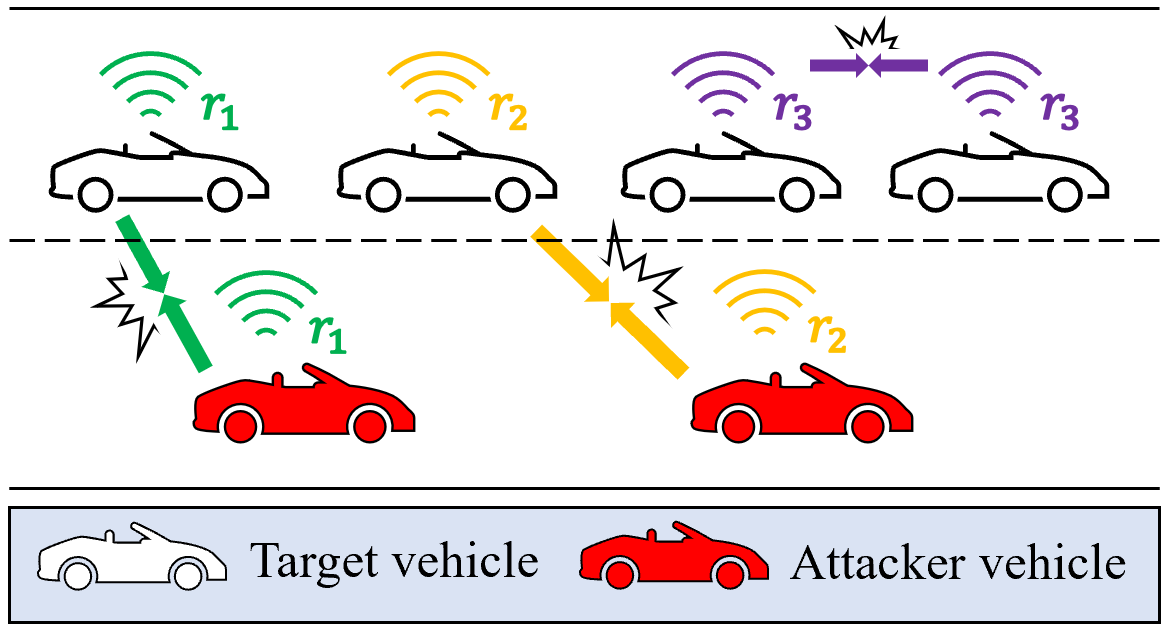}
\caption{Attack Type 3. The figure depicts an instance when the attacker vehicles have successfully chosen loner vehicles to attack, and have chosen different loner vehicles. This is because they cooperated in order to make the attack more effective.}
\label{fig:model3}
\end{figure}
\section{Analysis}
In this section, we analyze the potency of the different types of attacks. Lemmas~1 and~2 assume that the vehicles are synchronized implying that their semi-persistent periods are aligned.
\begin{lemma}
Consider attack type 1 (oblivious attack), with one attacker and one target. Then, for any $0<p \leq 1$, PRR is minimized with $p’=1$.
\end{lemma}

\begin{proof}
If there are two vehicles in the system, a single target and single attacker, then we can model these vehicles as a Markov chain with two possible states:
\begin{itemize}
    \item State 0: The two vehicles use different resource blocks.
    \item State 1: The two vehicles use the same resource block.
\end{itemize}
To minimize the PRR, the best strategy for the attacker is to maximize the fraction of time spent in state~1 because collisions take place in that state. The packet collision probability corresponds to the probability to find the Markov chain in state~1.
A Markov chain state diagram for this system is depicted in Figure~\ref{fig:mchain}.

\begin{figure}[t]
\centering
\begin{tikzpicture}
    \node[state] (statezero) {0};
    \node[state, right=of statezero] (stateone) {1};
    \draw[every loop,>=latex]
        (statezero) edge[bend right, auto=right]  node {$ \frac{p'}{N_r - 1}$} (stateone)
        (stateone) edge[bend right, auto=right] node {$1 - (1-p)(1-p')+\frac{pp'}{N_r-1}$} (statezero)
        (statezero) edge[loop left]             node {$1 - p' \frac{1}{N_r - 1}$} (statezero)
        (stateone) edge[loop right]             node {$(1-p)(1-p')+\frac{pp'}{N_r-1}$} (stateone);
\end{tikzpicture}
\caption{Markov chain state diagram for two vehicles (one attacker, one target). State 0 represents the state in which the target vehicle and the attacker vehicle use different resource blocks. State 1 represents the state in which the target and the attacker use the same resource block, resulting in packet collision.}
\label{fig:mchain}
\end{figure}
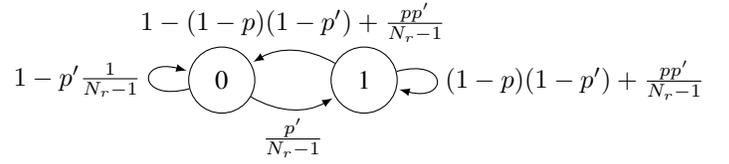

Next, we explain the transition probabilities shown in Figure~\ref{fig:mchain}. Recall that $N_r$ denotes the total number of resource blocks. A transition from state 0 to state 1 occurs when the attacker moves to a resource block used by the target vehicle. Note that the target vehicle will never switch to the resource block currently used by the attacker (because the channel is perceived as busy and the target won't select it as per the protocol). The transition probability from state 1 to state 0 is the complement of the transition probability from state 1 to itself. Such a self-transition occurs either if both the attacker and the target continue to use their current mutual resource block or if they both choose the same new resource block.

Let $\pi_0$ and $\pi_i$ respectively represent the stationary probabilities to find the Markov chain in state 0 and in state 1. These probabilities can be computed using the balance equation:
\begin{equation}\label{eq:balance}
    \pi_0 \frac{p'}{N_r-1} = \pi_1 \left( 1 - (1-p)(1-p')\frac{pp'}{N_r-1} \right).
\end{equation}
Combining Eq.~(\ref{eq:balance}) with the normalization equation $\pi_0 + \pi_1 =1$, we obtain
\begin{equation} \label{eq:pi_1}
 \pi_1 = \frac{p'}{(N_r-1)[p + p'(1-p)+\frac{p'}{N_r-1}(1-p)]}.
    \end{equation}

We next take the derivative of $\pi_1$ with respect to $p'$. If it is positive for all $p' \in [0,1]$, it means that $\pi_1$ is maximized when $p'=1$.

Taking the derivative of the expression given in Eq.~(\ref{eq:pi_1}) we obtain
\begin{equation}\label{eq:derivative}
    \frac{d\pi_1}{dp'} = \frac{p}{(N_r-1)(p+p'(1-p)+\frac{p'}{N_r-1}(1-p))^2} > 0,
\end{equation}
for all $0\leq p' \leq 1$ and $0< p \leq 1$.

Therefore, we conclude that $\pi_1$ increases with $p'$ and achieves its maximum when $p' = 1$. Conversely, the PRR decreases with $p'$ and is minimized when $p'=1$.
\end{proof}

From the proof, we note that the derivative of the packet collision probability $\pi_1$ decreases with $p’$ (while staying positive), i.e., increasing $p’$ leads to a diminishing return. That is, increasing $p'$ has a much greater effect in reducing the PRR while $p'$ is low compared to when $p'$ is already on the higher end of its range from $0$ to $1$. Our simulations in Section V show that the same effect holds true in the general case of many attackers and many targets.

From the proof, we further note that the derivative of the packet collision probability converges to 0 as $p$ tends to 0, that is the effect of $p'$ becomes negligible as $p$ tends to 0. The next lemma shows that this result holds in the general case.

\begin{lemma}
Consider attack type 1 (oblivious attack), with an arbitrary number of attackers $N_a$ and an arbitrary number of  target vehicles $N_v$. If $p=0$, then PRR is unaffected by $p’$ for any $0<p’ \leq 1$.
\end{lemma}
 
\begin{proof}
If $p=0$, then each target vehicle always stays on the same resource block (RB). Because the attackers choose their RBs randomly, in steady-state, each attacker is equally likely to be transmitting on any RB, independently of $p’$ (assuming $p’>0$, so that the initial state of each attacker does not have an effect in steady-state). Therefore, if there are $N_r$ RBs, the probability that a given attacker transmits on RB $r_i$ is $1/N_r$, for any $i \in \{1,2,\ldots,N_r\}$. Assume RB $r_i$ is used by a target vehicle,  then the steady-state probability that RB $r_i$ is jammed by one or more attackers  is $1-(1-1/N_r)^{N_a}$, which is independent of $p’$ and $i$. 
\end{proof}

While modifying $p'$ has no effect when $p=0$, in such situations, target vehicles cannot avoid using the same RB as an attacker. In order to take advantage of the sensing mechanism, it is preferable to set $p$ to a small value that is strictly greater than 0, as shown in our simulations in Section~V.

The next lemma states the intuitive result that a cooperative attack maximizes the number of loner vehicles being jammed in each time slot. Hence, this attack is expected to lead to the worst-case PRR.

\begin{lemma}
Consider attack type 3 (cooperative attack), where the attack period is $T_a = T_{tr}$, i.e., the attackers can switch RBs after each transmission period. For a given number of attackers $N_a$, this attack type maximizes the number of loner vehicles being jammed in each time slot.
\label{lemma:Lemma 3}
\end{lemma}

\begin{proof}
Assume that in any given time slot, there are  $N_l$ loner vehicles using RBs. The maximum number of loner vehicles that can be jammed is $N_j =\min(N_l,N_a)$, and attack type 3 will successfully target all $N_j$ vehicles in each time slot. Note that under attack type 2, this result is achieved only when there is a single attacker. With two or more attackers, there is a positive probability that the attackers will choose the same target.
\end{proof}

\section{Simulation Results}
In this section, we present Monte-Carlo simulation results to measure the potency of each of our attack types. Specifically, we aim to answer the following questions:
\begin{enumerate}
\item What impact does the number of attackers have on the potency of each type?
\item What impact does the value of $p'$ have on the potency of attack type 1?
\item Which attack types are the most and the least potent, and under what conditions?
\item What impact does the attack period have on the potency of each type? 
\end{enumerate} 

We measure the potency of the different attack types using the packet reception ratio (PRR) metric, which is the ratio of the number of packets that were successfully received to the total number of packets sent (BSMs in our case). We explore how PRR changes as a function of the total number of vehicles in the network, where the total number of vehicles is equal to the sum of the target vehicles and the attacker vehicles.

\subsection*{Simulation Setup}
To simulate the attack types, we adapt a Monte-Carlo simulation model in Matlab of C-V2X networks that was introduced in \cite{Wang2018}. The simulations are run under the following assumptions and conditions unless otherwise noted in the figures' captions:

\begin{itemize}
    \item The vehicular network is fully-connected, implying that all vehicles are within the range of one another.
    \item There is no signal fading and packet losses are only due to collisions.
    \item Simulation time is 300 s.
    \item Number of simulation trials is 10.
    \item Transmission period, $T_{tr}$, is 100 ms.
    \item The length of the semi-persistent period is $T_sT_{tr} = 1$ s. 
    \item The attack period is set equal to the semi-persistent period (i.e., the attacks are deniable), except for the results presented in Section~\ref{sec:attack_period}. 
    \item There are 200 resource blocks ($N_r = 200$).
    \item The probability that a target vehicle change its resource block is $p = 0.2$.\\
\end{itemize}

Note that the initial offsets of the semi-persistent periods for the different vehicles are chosen at random, and hence the semi-persistent periods are not synchronized.

In our investigation, we first zoom in on individual attack types and plot the utilization of resource blocks over time with five attackers and five target vehicles, and the attack period set equal to the semi-persistent period. We next proceed to explore the PRR for each attack type by varying the total number of vehicles. We compare the PRR performance of the baseline case with the performance of our attack types when the number of attackers, $N_a$, is equal to $1$ and $5$. We also explore the effect of varying the attackers' $p'$ in attack type 1, as well as the effect of the attack period, $T_a$, on all three types. Moreover, we analyze the potency of our three attack types against one another.
 
\subsection{Attack Type 1: Oblivious Attack} 
In this attack type, attacker vehicles choose random resource blocks with the probability $p'$. We focused on two special cases: $p'=0$ and $p'=1$. Figures \ref{fig:greenred_type1_p'0} and \ref{fig:greenred_type1_p'1} show the corresponding time-RB utilization plots. Vehicles 1 through 5 are target vehicles, and vehicles 6 through 10 are attacker vehicles. As can be seen from the utilization plots, there are more collisions when $p' = 1$ as expected. For the PRR performance plots, we vary the total number of vehicles, as well as the number of attackers, $N_a$. The results are shown in Figure \ref{fig:PRR_model1_scP0} and Figure \ref{fig:PRR_model1_scP1} for $p'=0$ and $p'=1$, respectively. The blue curve is the PRR of the baseline representing the case when there are no attackers. As can be seen, the PRR drops with the increase of the total number of vehicles in both special cases. The special case of $p'=0$ appears to have very little to no effect irrespective of the number of attackers. The special case of $p'=1$ drops the PRR for a single attacker and five attackers by approximately 0.5\% and 2\% on average respectively. This matches the analysis provided by Lemma 1 that showed that PRR is minimized when $p'=1$ in attack type~1, but to further confirm that, we next investigate the effect of varying $p'$ in attack type 1.  

\begin{figure}[t]
\centering
\includegraphics[scale=0.85]{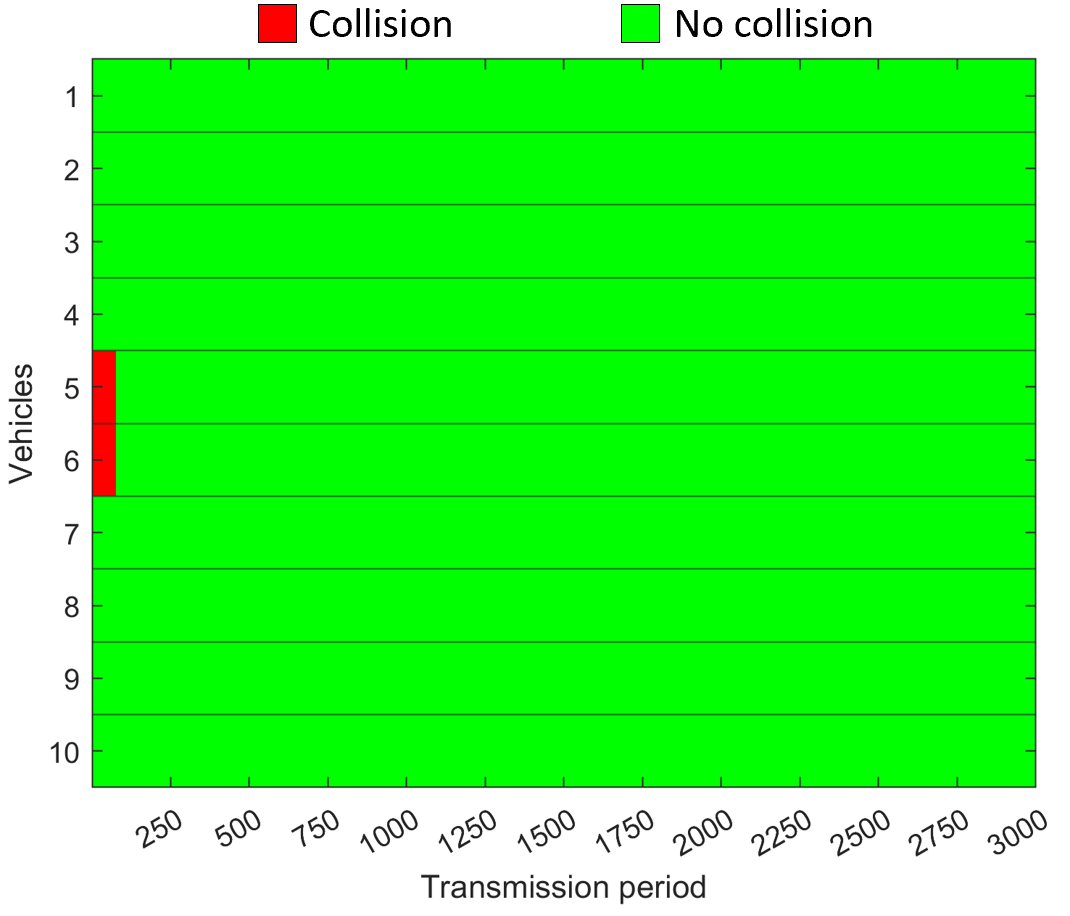}
\caption{Time-RB utilization plot for \textbf{attack type 1}, $p'=0$,  and the attack period is set to the semi-persistent period. Vehicles 1 through 5 are target vehicles, and vehicles 6 through 10 are attacker vehicles. In this instance, target vehicle 5 and attacker vehicle 6 collide at the beginning because they happen to choose the same resource block at the same time. This results in a collision that lasts until the target changes its resource block. No other collisions occur as the attacker vehicles do not change the resource blocks they selected at the beginning (due to $p'$ being $0$), and none of the target vehicles will choose an already occupied resource block.
}
\label{fig:greenred_type1_p'0}
\end{figure}

\begin{figure}[t]
\centering
\includegraphics[scale=0.85]{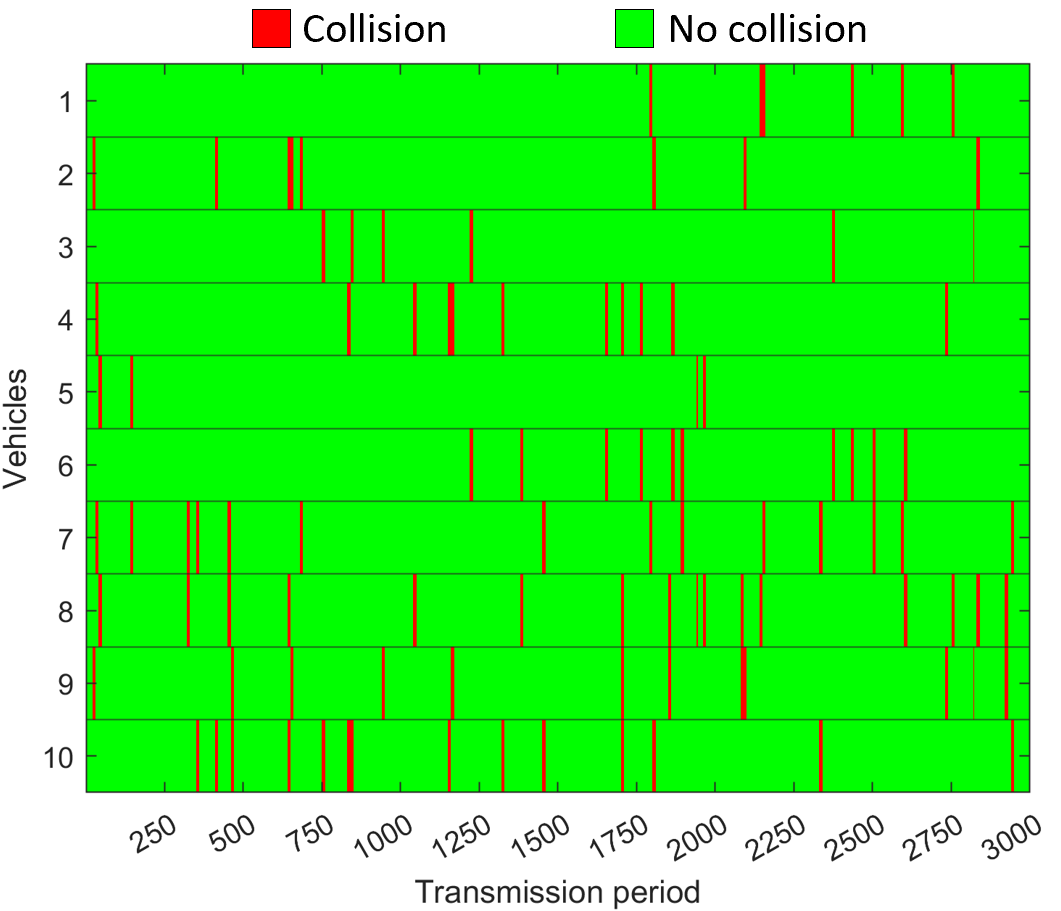}
\caption{Time-RB utilization plot for \textbf{attack type 1}, $p'=1$, and the attack period is set to the semi-persistent period. Vehicles 1 through 5 are target vehicles, and vehicles 6 through 10 are attacker vehicles. In this instance, various vehicles collide over time. Most collisions are due to targets and attackers randomly choosing the same RB at the same time. On occasion, two attackers will randomly choose the same RB as well. Collisions do not last long as attackers choose a new RB every semi-persistent period, but they do occur more frequently compared to the case when $p' = 0$.
}
\label{fig:greenred_type1_p'1}
\end{figure}

\begin{figure}[t]
\centering
\includegraphics[scale=0.20]{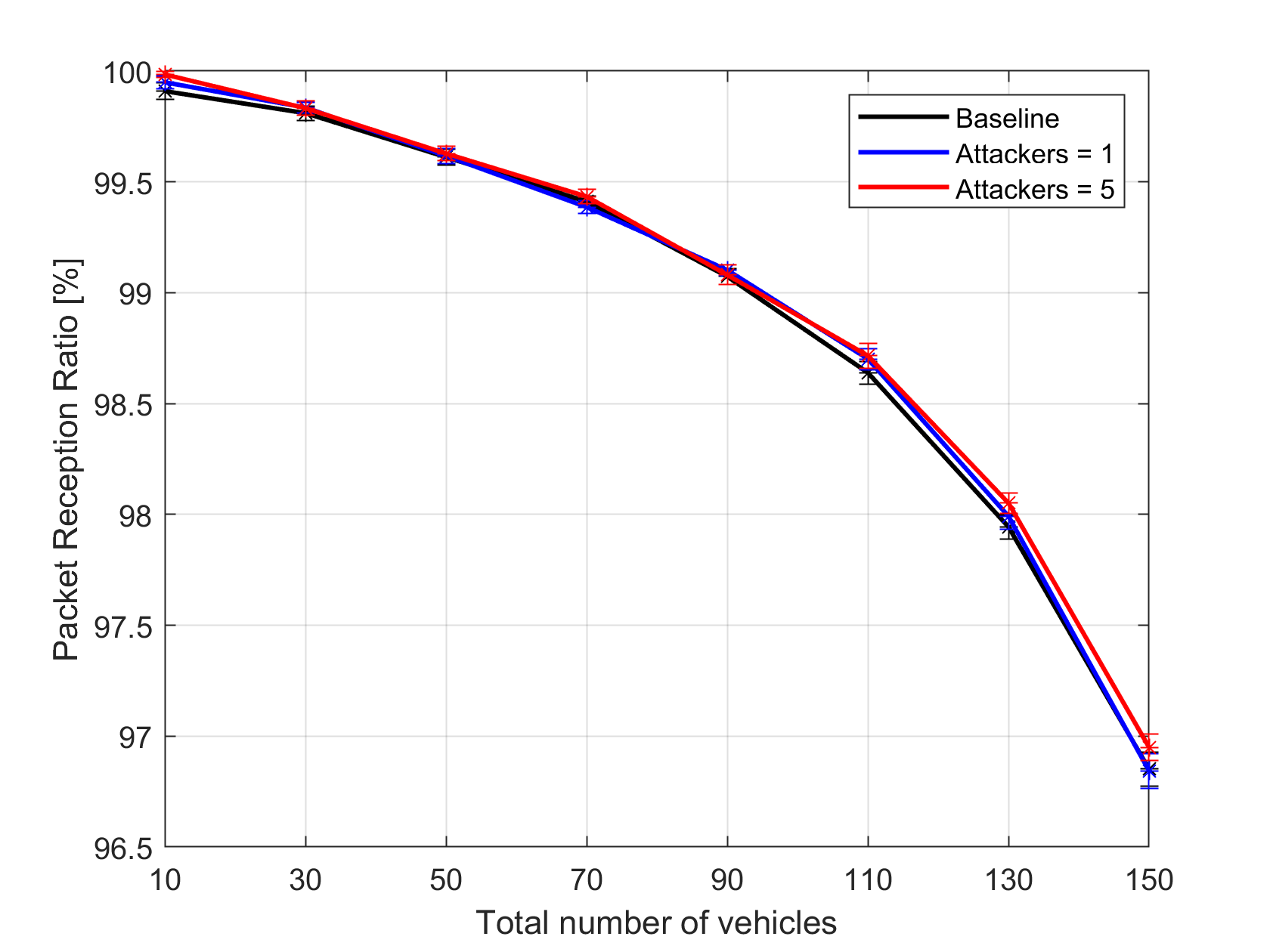}
\caption{PRR versus the total number of vehicles in \textbf{attack type 1, $p' = 0$} with $95\%$ confidence intervals. Simulation time is 3000 s.}
\label{fig:PRR_model1_scP0}
\end{figure}

\begin{figure}[t]
\centering
\includegraphics[scale=0.20]{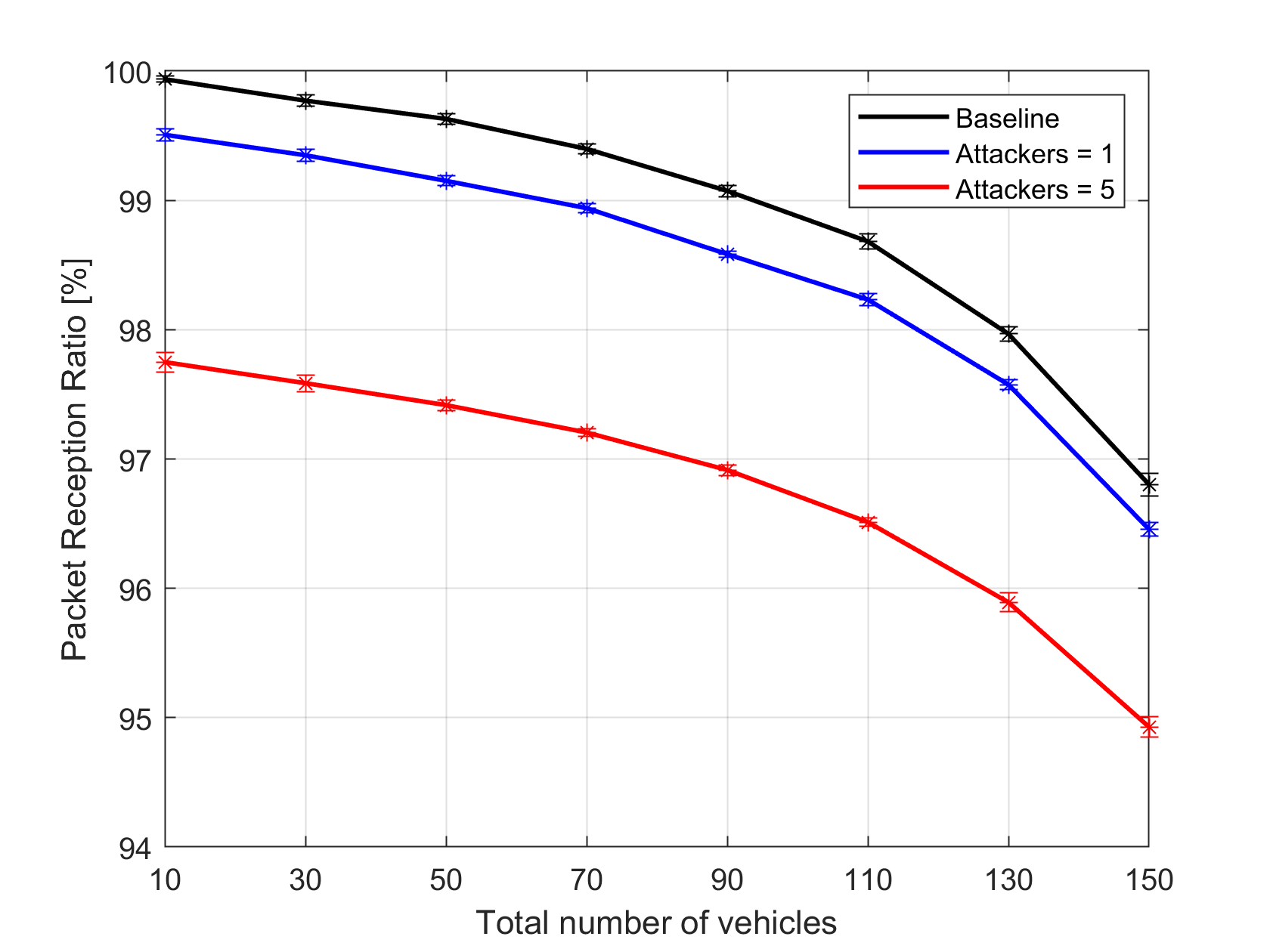}
\caption{PRR versus the total number of vehicles in \textbf{attack type 1, $p' = 1$} with $95\%$ confidence intervals. Simulation time is 3000 s.}
\label{fig:PRR_model1_scP1}
\end{figure}

\subsubsection*{Varying $p'$ in Attack Type 1}
To answer our second research question regarding the impact of the value of $p'$ on the potency of attack type 1, we evaluate the PRR against the total number of vehicles. We compare the performance of the attack type at $p'=0$, $p'=0.5$, and $p'=1$  while the number of attackers $N_a$ is fixed at $5$. As shown by the resulting Figure~\ref{fig:varyP'}, the PRR drops with the increase of $p'$ and $p'=1$ has the lowest PRR, in line with Lemma~1. This is perhaps the most obvious in Figure \ref{fig:varyP'_Nv} where we fixed the number of attackers, $N_a$, to five and observed the behavior of PRR as $p'$ is varied. The three resulting curves represent different number of target vehicles, $N_v$. As observed from our analysis in Section IV, the effect of increasing $p'$ becomes less pronounced as $p'$ gets larger.

\begin{figure}[t]
\centering
\includegraphics[scale=0.20]{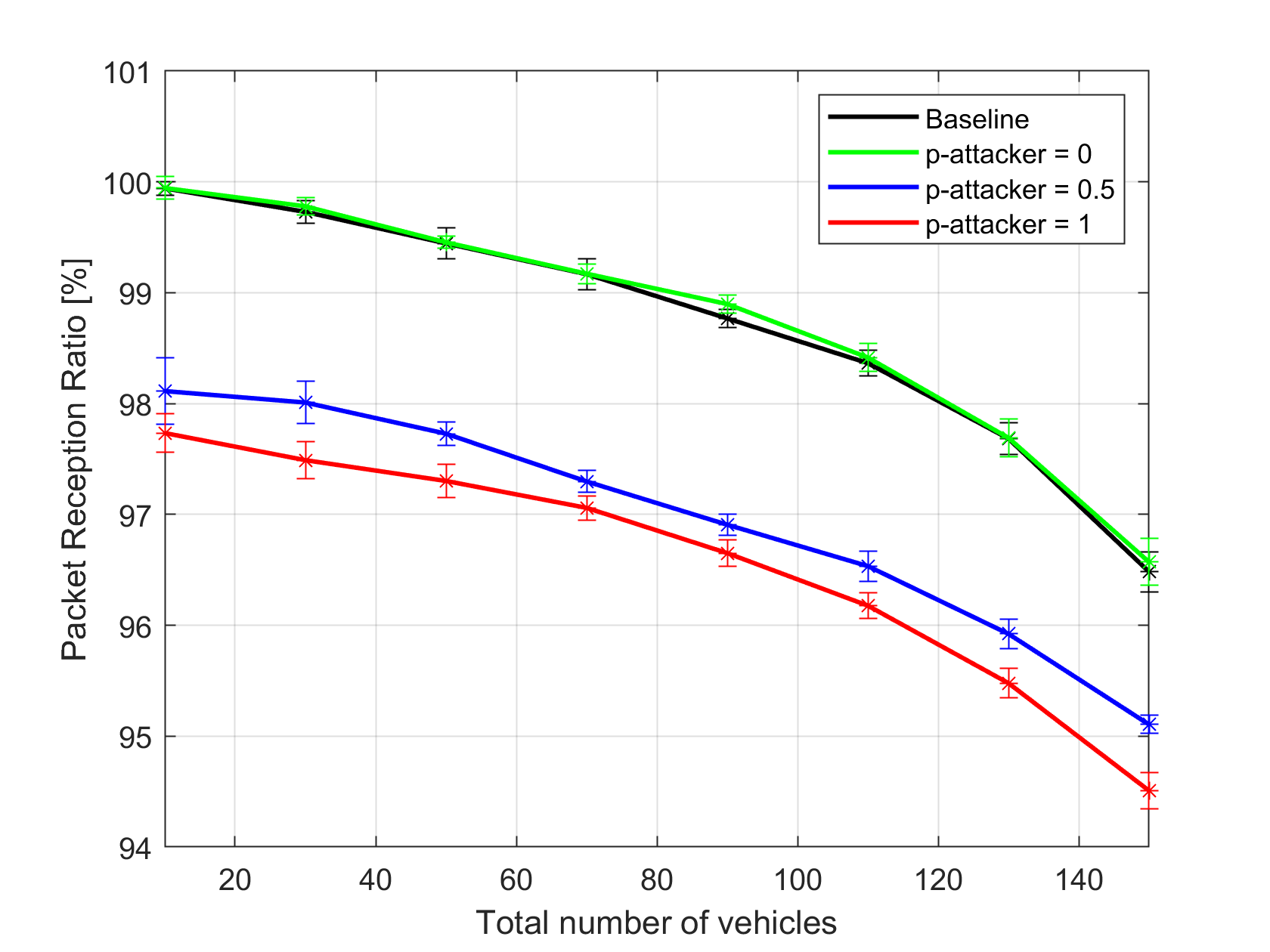}
\caption{PRR versus the total number of vehicles in \textbf{attack type 1} with $95\%$ confidence intervals. We vary the value of $p'$ to explore its effect on the potency of attack type 1. When $p'$ is equal to $0$, the attack has almost no effect. The PRR drops as the $p'$ is increased.}
\label{fig:varyP'}
\end{figure}

\begin{figure}[t]
\centering
\includegraphics[scale=0.20]{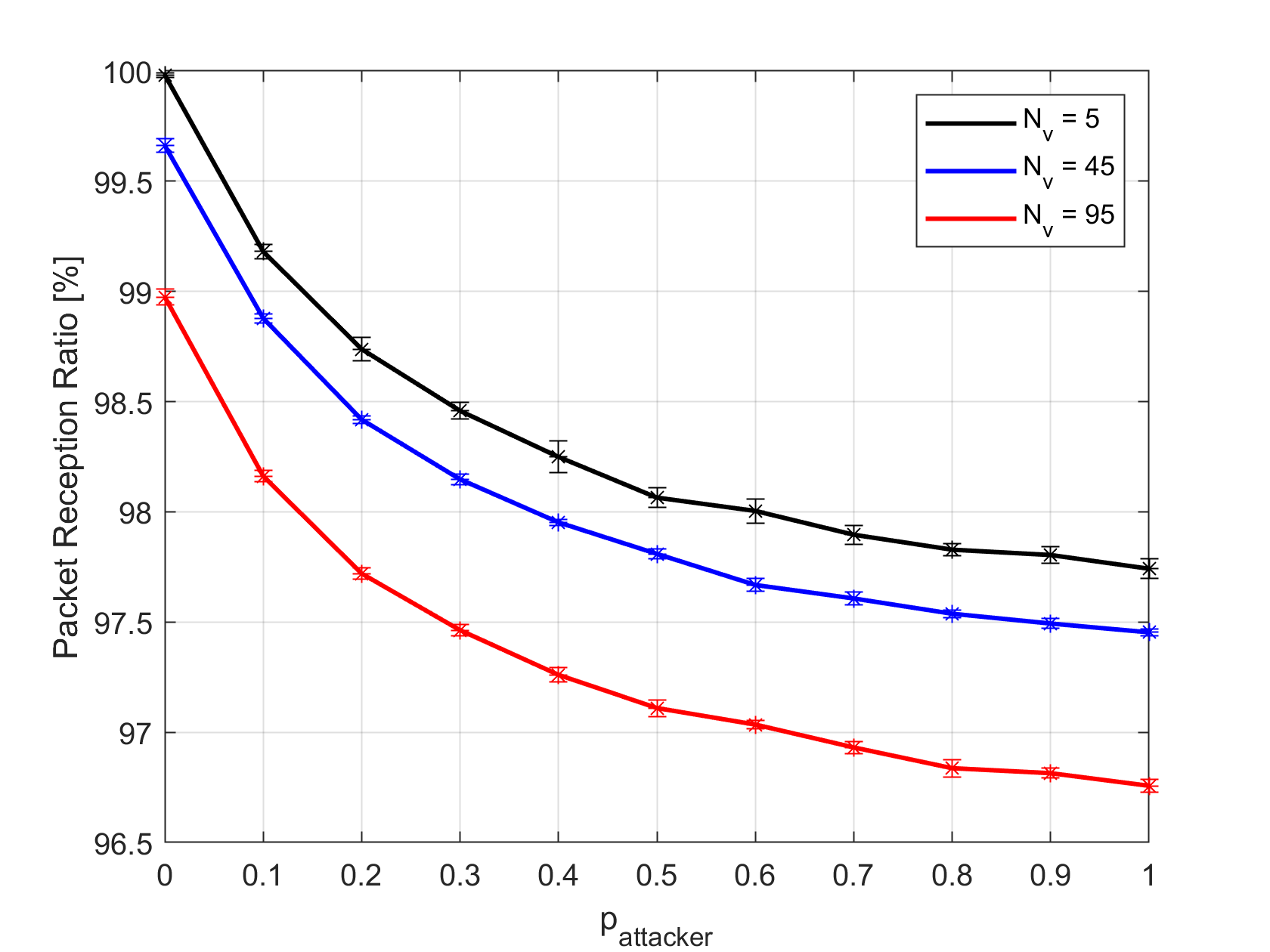}
\caption{PRR versus $p'$ in \textbf{attack type 1} with $95\%$ confidence intervals and different values of the target vehicles, $N_v$. The number of attackers, $N_a$, is fixed at 5. As $p'$ increases, the PRR drops, with $p'=1$ having the lowest PRR for any number of target vehicles. Simulation time is 30000 s.}
\label{fig:varyP'_Nv}
\end{figure}

\subsubsection*{Varying $p$ in Attack Type 1}
We next consider the effect of target vehicles changing the value of $p$. We consider a scenario with $N_v = 45$ target vehicles, $N_a = 5$ attackers, and $p'=0.5$. 

Figure~\ref{fig:varyP_Nv} displays the results. As expected from the discussion following Lemma~2, choosing a small value for $p$, in this case around 0.3, strikes the best balance between minimizing the impact of attackers and avoiding collisions with other target vehicles.

\begin{figure}[t]
\centering
\includegraphics[scale=0.2]{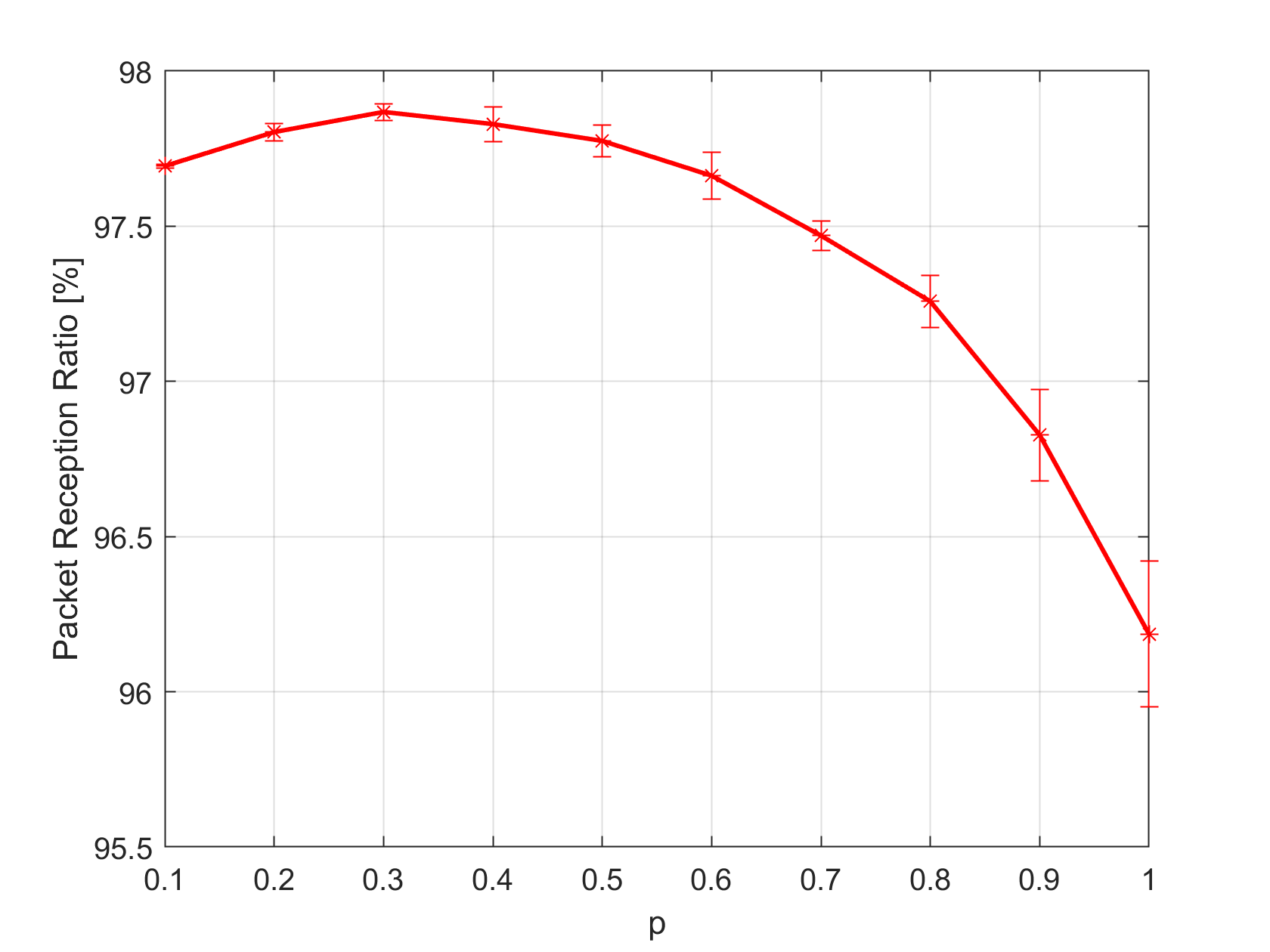}
\caption{PRR versus $p$ in \textbf{attack type 1} with $95\%$ confidence intervals. The number of 
target vehicles is fixed at $N_v = 45$, the number of attackers is fixed at $N_a = 5$, and $p'=0.5$. Simulation time is 120000 s.}
\label{fig:varyP_Nv}
\end{figure} 

\subsection{Attack Type 2: Smart Attack}
In this attack type, attacker vehicles look for resource blocks that are being used by loner vehicles and randomly select one of those resource blocks. Attackers do not cooperate implying that they may choose the same loner vehicle to target during one attack period. Figure \ref{fig:greenred_type2} shows the time-RB utilization plot for attack type 2. As can be seen, more collisions occur than in either special case of attack type 1. We also observe the PRR as the total number of vehicles and the number of attackers vary. The results of this attack type are shown in Figure \ref{fig:PRR_model2}. The blue curve is the baseline representing PRR when there are no attackers. When the total number of vehicles is 10, this attack type significantly lowers the PRR, down to 90\%  for a single attacker and to 55\% for 5 attackers. The effectiveness of this attack type reduces as the total number of target vehicles increases, however.

\begin{figure}[t]
\centering
\includegraphics[scale=0.85]{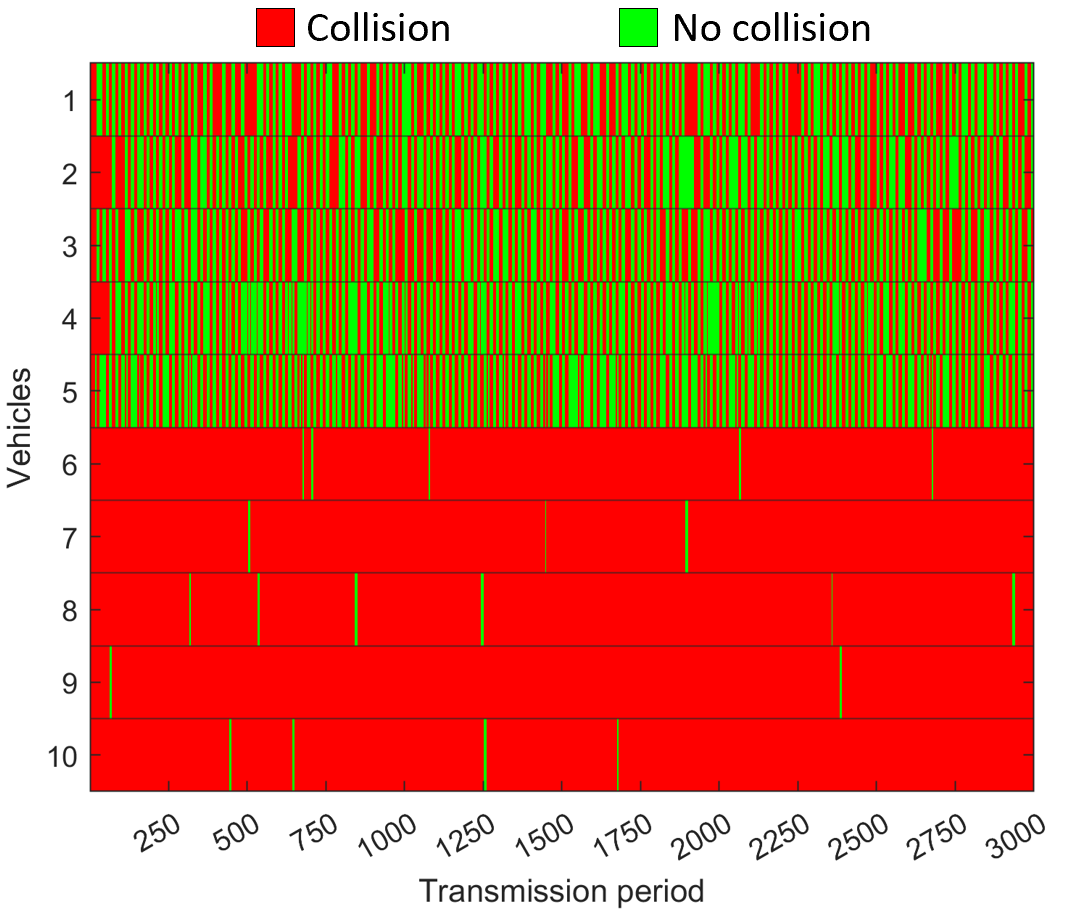}
\caption{Time-RB utilization plot for \textbf{attack type 2}, and the attack period is set to the semi-persistent period. Vehicles 1 through 5 are target vehicles, and vehicles 6 through 10 are attacker vehicles. In this instance, collisions occur in every semi-persistent period. However, not all target vehicles are under attack in each semi-persistent period because attackers do not cooperate and sometimes choose the same target. Additionally, not all attackers are always in collision because their targets sometimes change their RB during an on-going attack. This happens because targets' and attackers' semi-persistent periods are not always aligned.
}
\label{fig:greenred_type2}
\end{figure}

\begin{figure}[t]
\centering
\includegraphics[scale=0.20]{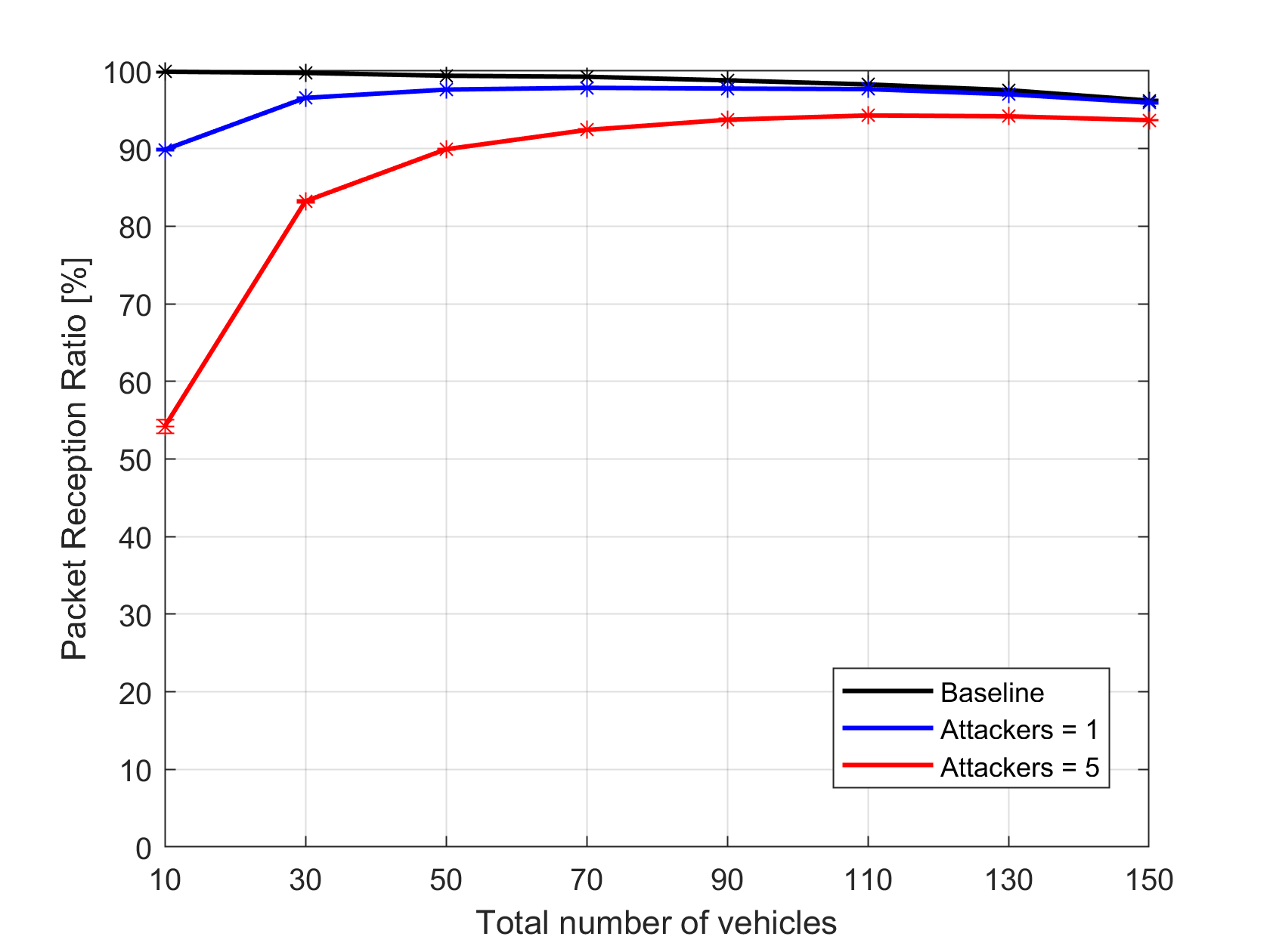}
\caption{PRR versus the total number of vehicles in \textbf{attack type 2} with $95\%$ confidence intervals. }
\label{fig:PRR_model2}
\end{figure}

\subsection{Attack Type 3: Cooperative Attack}
In this attack type, just like in attack type 2, attacker vehicles look for resource blocks that are being used by loner vehicles. In this type however, attacker vehicles cooperate and ensure that they all choose a different target. Attack type 3 corresponding time-RB utilization plot is shown in Figure \ref{fig:greenred_type3}. As can be seen, more collisions occur in attack type 3 compared to attack type 2. We also observe the PRR as the total number of vehicles and the number of attackers vary. The results of this attack type are shown in Figure \ref{fig:PRR_model3}. The blue curve is the baseline representing PRR when there are no attackers. As can be seen from the figure, this attack is highly efficient (PRR = 10\%) when there are $5$ attackers and the total number of vehicles is $10$. This is because there are $5$ targets and $5$ attackers that collaborate, which results in collisions on all 5 targets. The PRR at that point is not $0$ because
the attacks occur every semi-persistent period and the target vehicles are not synchronized with respect to it. Hence, there will be target vehicles that will change their resource blocks during an attack period and before the attacker has a chance to react which effectively reduces the number of collisions. 

\begin{figure}[t]
\centering
\includegraphics[scale=0.85]{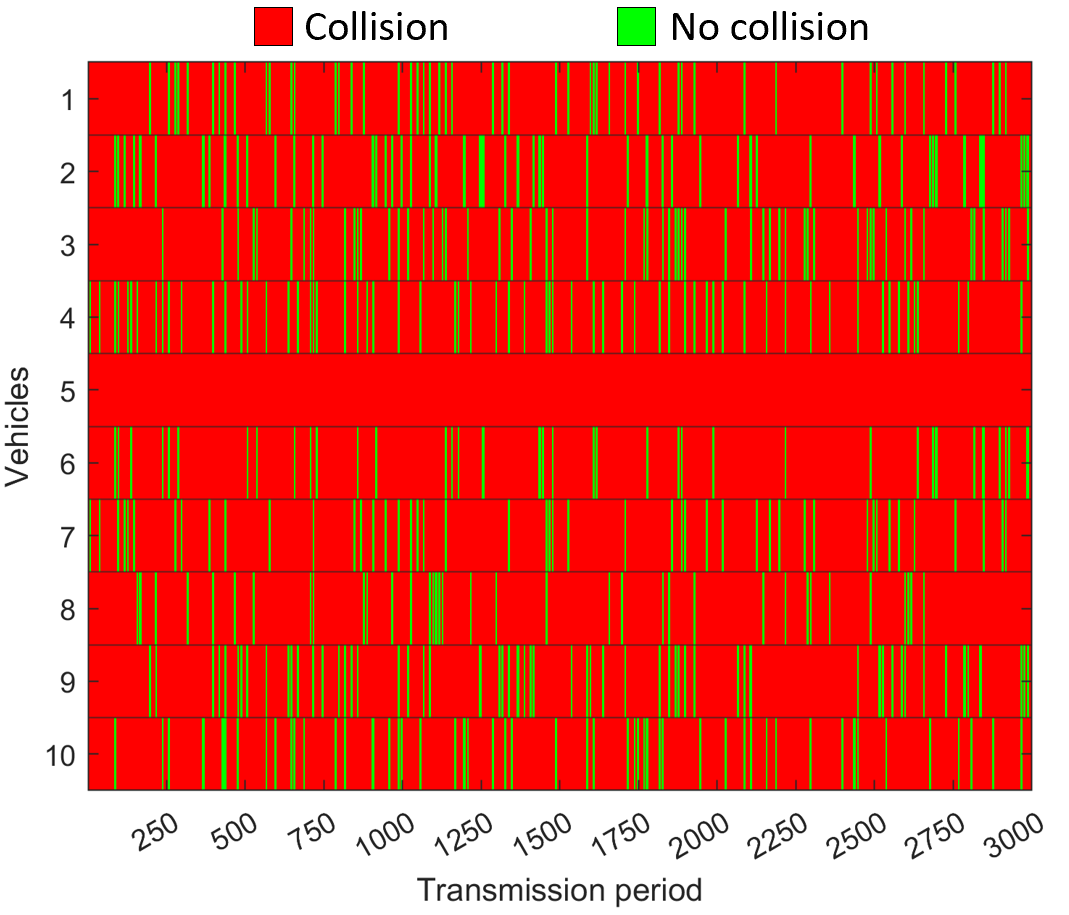}
\caption{Time-RB utilization plot for \textbf{attack type 3}, and the attack period is set to the semi-persistent period. Vehicles 1 through 5 are target vehicles, and vehicles 6 through 10 are attacker vehicles. In this instance, target vehicles spend most of the time colliding with attacker vehicles because attackers cooperate in choosing their targets. Target vehicles will sometimes change their RBs during an attack period, which is why there are some successful transmission i.e. no-collision slots.
}
\label{fig:greenred_type3}
\end{figure}

\begin{figure}[t]
\centering
\includegraphics[scale=0.20]{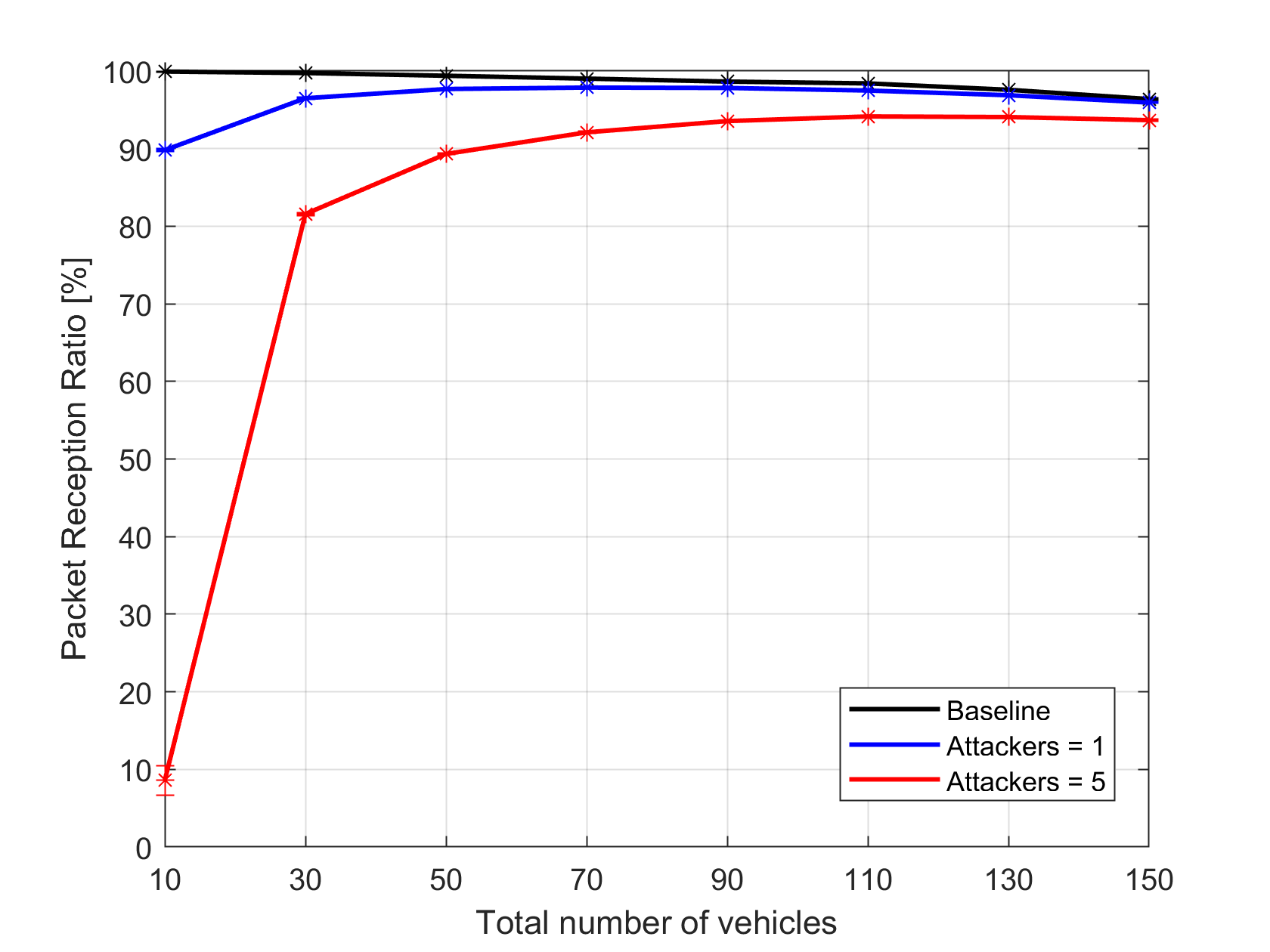}
\caption{PRR versus the total number of vehicles in \textbf{attack type 3} with $95\%$ confidence intervals.}
\label{fig:PRR_model3}
\end{figure}

\subsection{Comparison between Attack Types}
To answer our third research question, we plot the PRR versus the total number of vehicles for $N_a = 5$ for all the attack types on the same graph (see Figure \ref{fig:ALLmodels}). As can be seen from the figure, attack type 3 is the most effective at minimizing the PRR when the vehicle density is low. Its effect is reduced as the total number of vehicles increases. This is because there will be relatively a smaller fraction of packet collisions caused by the attackers.
Attack type 1 is the least effective in minimizing the PRR of the three, but the other two attack types perform only slightly better at high vehicle density.

\begin{figure}[t]
\centering
\includegraphics[scale=0.20]{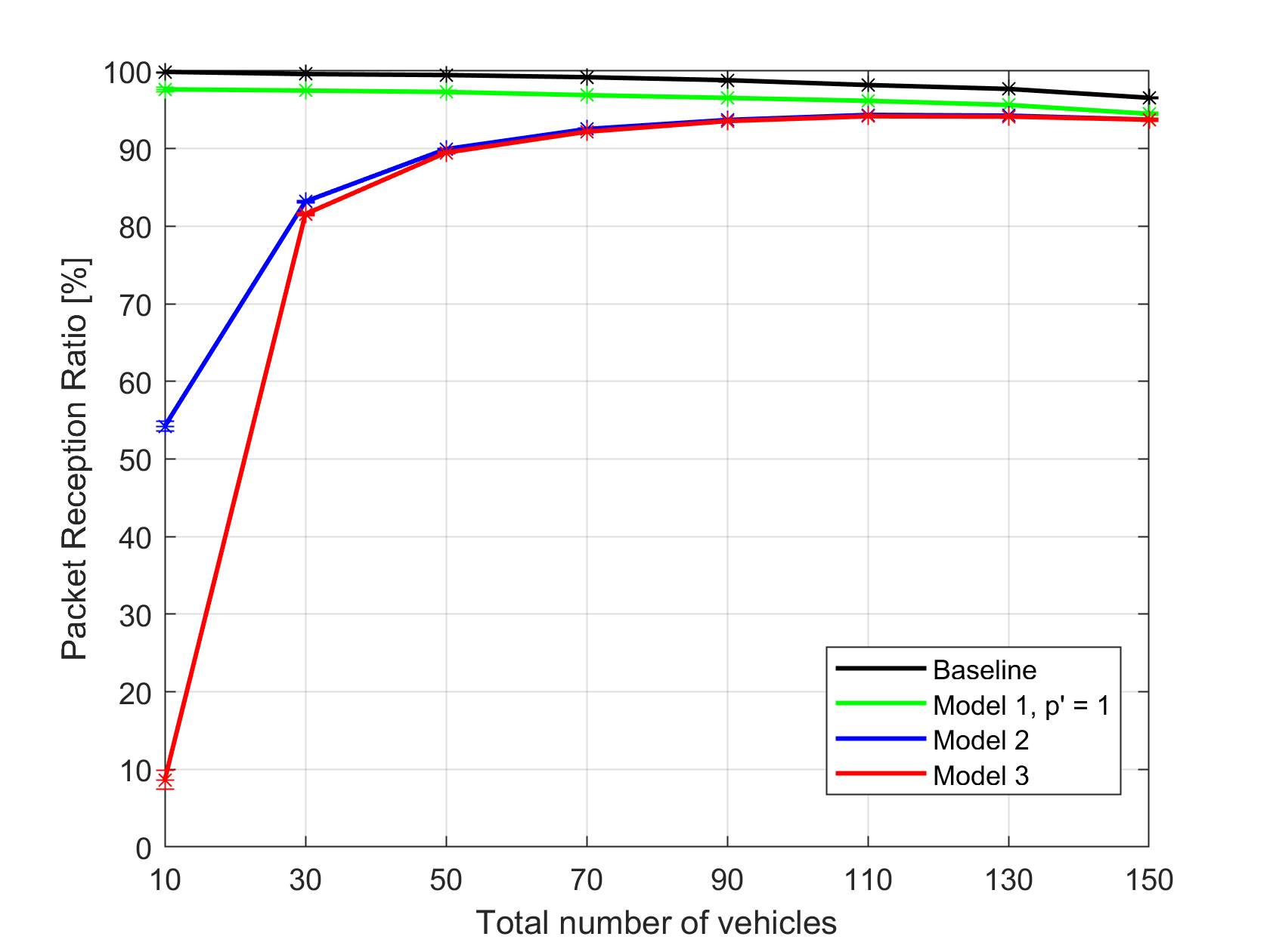}
\caption{Comparison of the attack types. PRR versus the total number of vehicles with $95\%$ confidence intervals, $N_a = 5$.}
\label{fig:ALLmodels}
\end{figure}

\subsection{Deniability: Varying the Attack Period}
\label{sec:attack_period}
To answer our last research question regarding the impact of the attack period, $T_a$, on the potency of each type, we set a different value of the attack period as compared to the default attack period. 
Recall that our default attack period was set equal to the duration of a semi-persistent period. This was chosen because such an attack is considered deniable and therefore less likely to be caught. We next evaluate the effect of setting the length of the attack period equal to the transmission period, $T_a = T_{tr}$.  
Such an attack would be considered a non-deniable attack. The time-RB utilization plot for the non-deniable attack type 3 with five attackers is shown in Figure \ref{fig:greenred_type3_ttr}. As can be seen, all target vehicles are colliding with the attackers throughout the duration of the simulation as expected. 

\begin{figure}[t]
\centering
\includegraphics[scale=0.85]{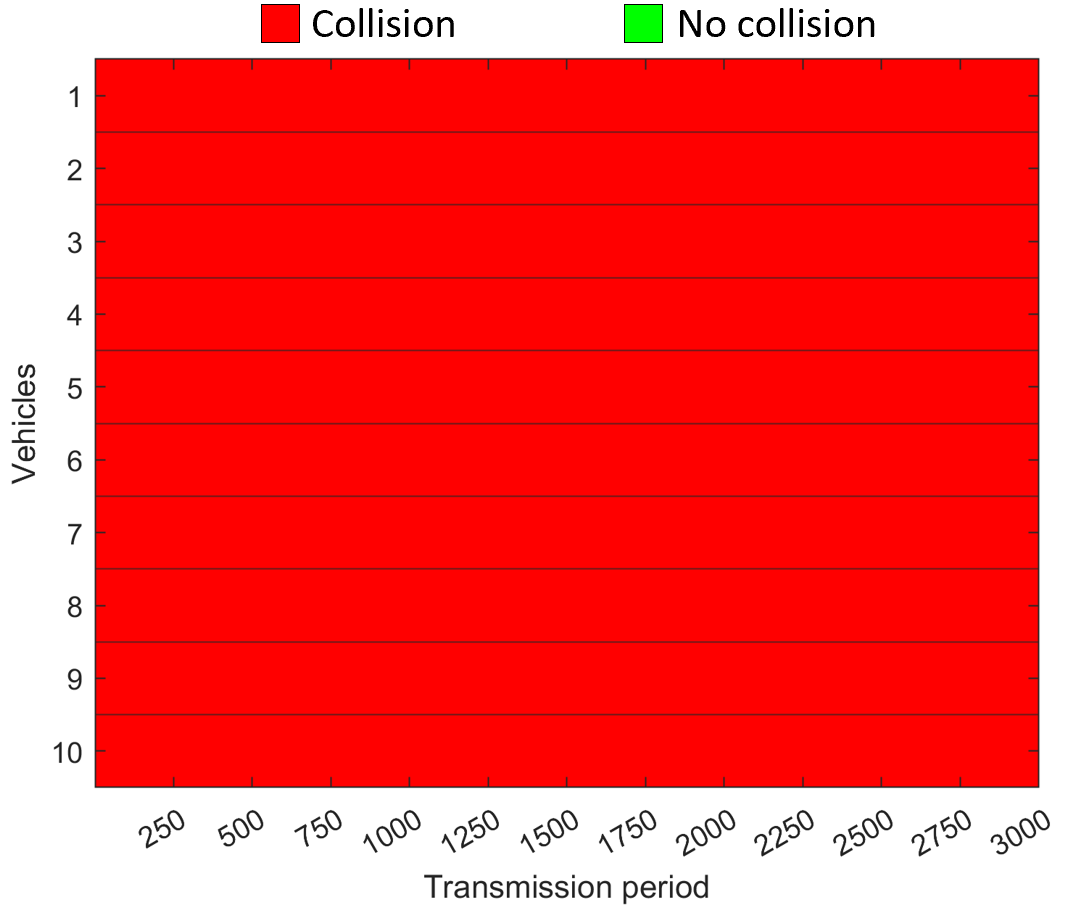}
\caption{Time-RB utilization plot for non-deniable \textbf{attack type 3}, and the attack period is set to the transmission period. Vehicles 1 through 5 are target vehicles, and vehicles 6 through 10 are attacker vehicles. In this instance, target vehicles are always colliding with the attacker vehicles because attackers cooperate and they attack every transmission period. Hence, the non-deniable attack type 3 is the most potent attack type when there are five targets and five attackers.}
\label{fig:greenred_type3_ttr}
\end{figure}

\begin{figure}[t]
\centering
\includegraphics[scale=0.2]{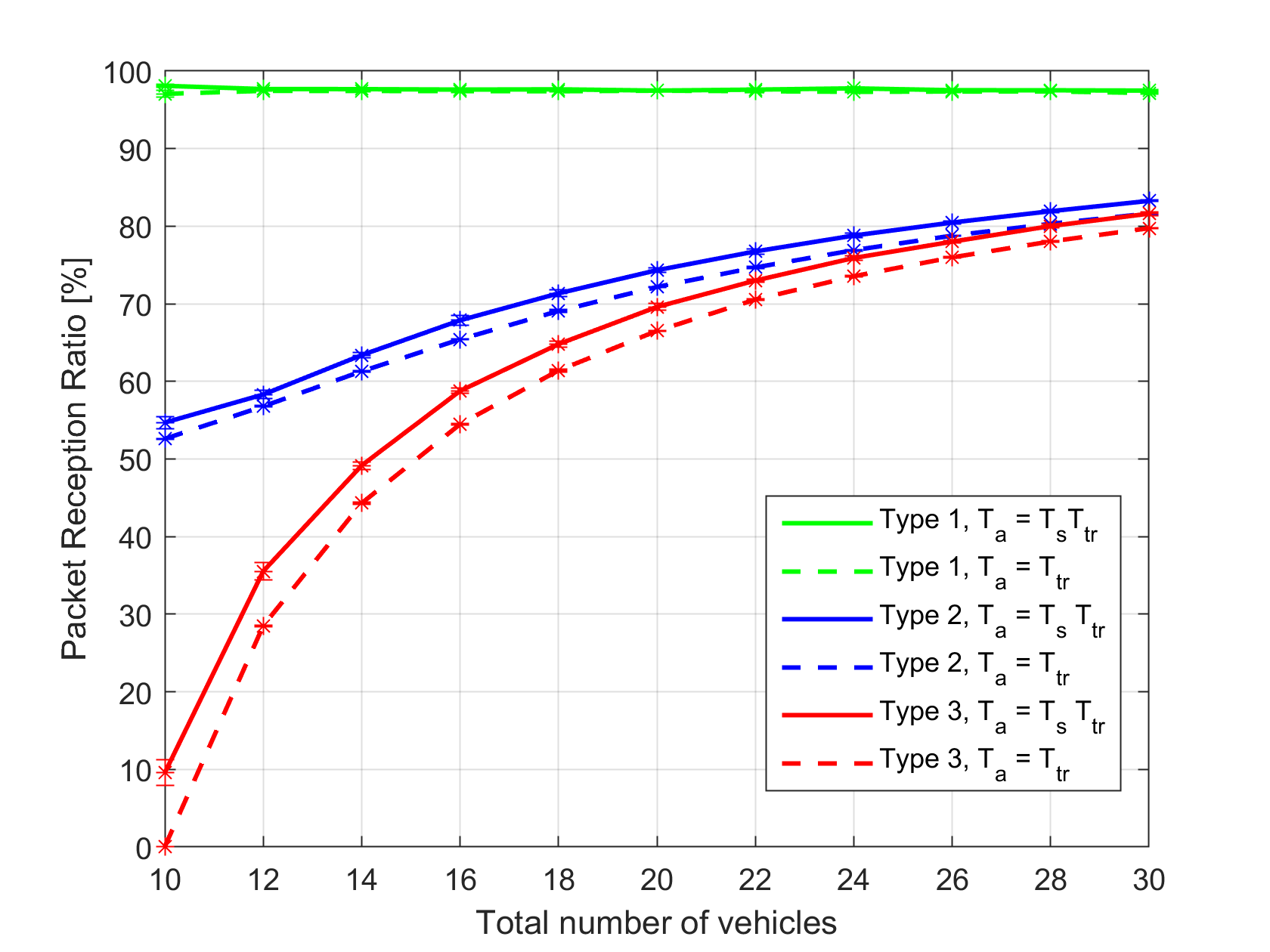}
\caption{Comparing all types' potency with different attack periods, low vehicle density. PRR versus the total number of vehicles with $95\%$ confidence intervals, $N_a = 5$.}
\label{fig:ALLmodels_low}
\end{figure}

\begin{figure}[t]
\centering
\includegraphics[scale=0.2]{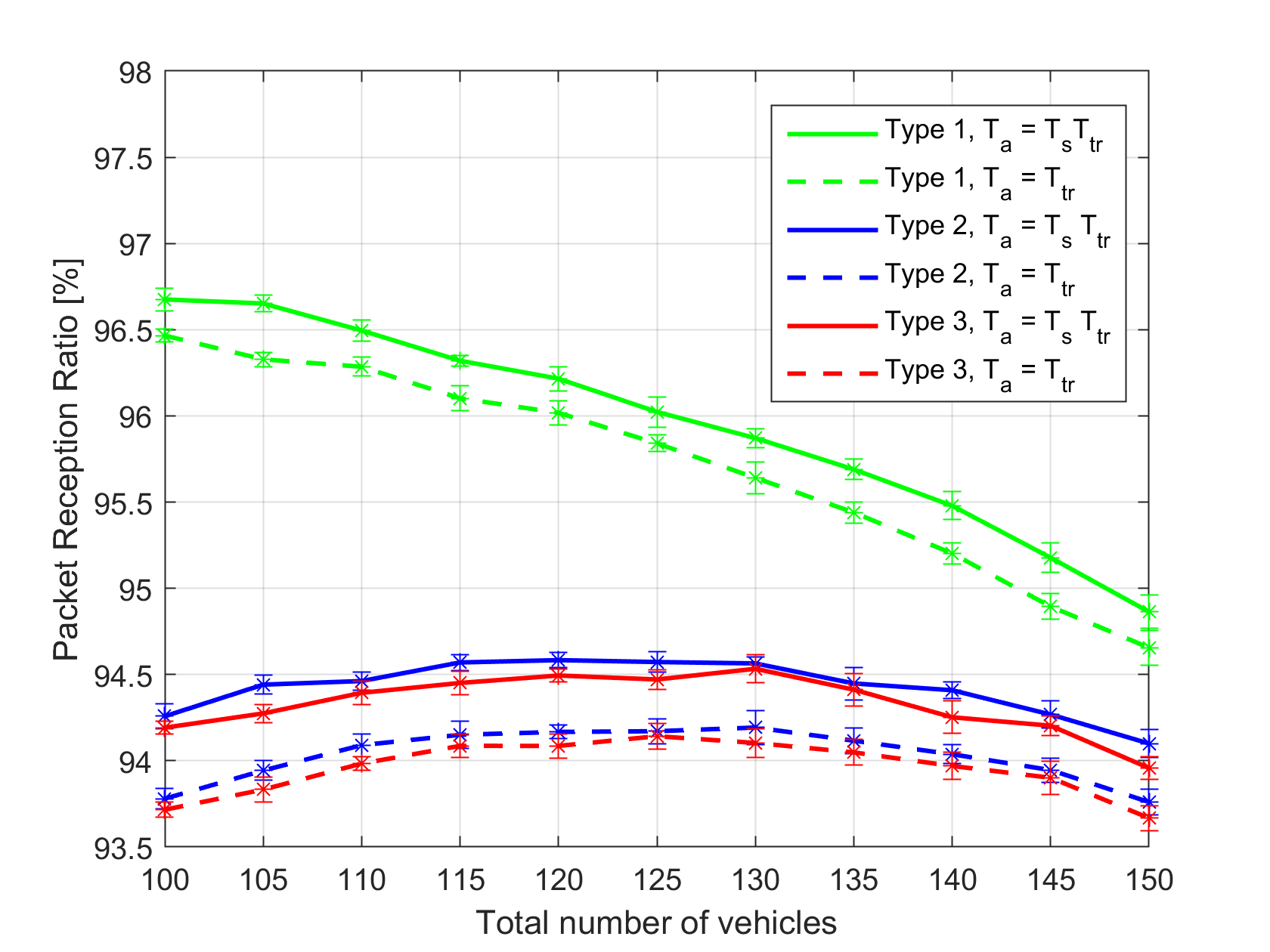}
\caption{Comparing all types' effectiveness with different attack periods, high vehicle density. PRR versus the total number of vehicles with $95\%$ confidence intervals, $N_a = 5$, simulation time is 3000 s.}
\label{fig:ALLmodels_high}
\end{figure}

Next, we fix the number of attackers, $N_a$, to 5 to observe the PRR performance of various attack types for different attack periods. For attack type 1, we set $p' = 1$. The results can be observed in Figures \ref{fig:ALLmodels_low} and \ref{fig:ALLmodels_high}.
As expected, non-deniable attacks ($T_a = T_{tr}$) of each attack type are more effective (have lower PRR) compared to their deniable counterparts ($T_a = T_sT_{tr}$). The PRR under non-deniable attack type 3 is the worst-possible, as expected from  Lemma \ref{lemma:Lemma 3}. We note as the total number of vehicles grows, attack type 1 becomes almost as effective as attack types 2 and 3. Finally, at high vehicle density, we note that the non-deniable attacks are only slightly more effective than the deniable attacks.




\section{Conclusion}
We introduced three types of denial-of-service attacks on C-V2X networks operating in Mode 4: oblivious, smart, and cooperative attacks. We performed extensive Monte-Carlo simulation, complemented with theoretical analysis, to investigate the potency of each attack type. We gained the following insights from our investigation: (1) the oblivious attack is most effective when $p'=1$; (2) for a fixed number of attackers, the collaborative and smart attacks are most effective at low vehicle density and have a significant impact. When the number of target vehicles grows, oblivious attacks become almost as effective; (3) switching resource block every transmission period leads only to minor gain over switching resource block every semi-persistent period, the latter attack being more easily deniable.

Our work opens several directions for further research, including modeling vehicular mobility and limited communication range, evaluating other attack objectives that are different from minimizing PRR, and incorporating channel impairments. 
On the analytical side, formally proving that $p'=1$ is optimal for an oblivious attacker in the general case is an interesting open question. 


\section*{Acknowledgment}
This research was supported in part by NSF under grants CNS-1908087 and CNS-1908807.
\bibliographystyle{IEEEtran}

\end{document}